\documentclass[letterpaper,12pt]{amsart}
\usepackage{latexsym,array,delarray,amsthm,amssymb,epsfig}

\theoremstyle{plain}
\newtheorem{thm}{Theorem}[section]
\newtheorem{lemma}[thm]{Lemma}
\newtheorem{prop}[thm]{Proposition}
\newtheorem{cor}[thm]{Corollary}

\newtheorem*{thm*}{Theorem}
\newtheorem*{lemma*}{Lemma}
\newtheorem*{prop*}{Proposition}
\newtheorem*{cor*}{Corollary}
\newtheorem*{conj*}{Conjecture}

\theoremstyle{definition}
\newtheorem{defn}[thm]{Definition}
\newtheorem{ex}[thm]{Example}

\newtheorem{alg}[thm]{Algorithm}

\theoremstyle{remark}

\newcommand{\rr}{\mathbb{R}}

\newcommand{\ind}{\mbox{$\perp \kern-5.5pt \perp$}}

\begin{document}

\title{Polyhedral Combinatorics of UPGMA Cones}
\author{Ruth Davidson}
\author{Seth Sullivant}
\maketitle

\begin{abstract}
Distance-based methods such as UPGMA (Unweighted Pair Group Method 
with Arithmetic Mean)  continue to play a significant role in 
phylogenetic research. We use polyhedral combinatorics to analyze 
the natural subdivision of the positive orthant induced by 
classifying the input vectors according to tree topologies returned
by the algorithm. The partition lattice  informs the study
of UPGMA trees.   We give a closed form for the extreme rays of 
UPGMA cones on $n$ taxa, and compute the normalized volumes of 
the UPGMA cones for small $n$.
\end{abstract}


\section{Introduction}

The UPGMA algorithm (Unweighted Pair Group Method with Arithmetic Mean) 
\cite{Sokal1963} is an agglomerative
tree reconstruction method, that takes as input ${n \choose 2}$ 
pairwise distances (dissimilarities) 
between $n$ taxa and returns a rooted, equidistant tree with 
these $n$ taxa as the leaves.
UPGMA is a greedy heuristic that attempts to compute the Euclidean
projection onto the space of all equidistant tree metrics \cite{Fahey2008}.  
The UPGMA algorithm subdivides the positive orthant 
$\rr^{n(n-1)/2}_{\geq 0}$ into regions
based on which combinatorial type of tree is returned by the algorithm.
The goal of this paper is to study the geometry of these regions in order to 
understand both how the regions relate to one
 another as well as the performance
of the algorithm.

UPGMA has poor performance if the data is tree-like but does not follow
a molecular clock.  In spite of this limitation, we find UPGMA
an interesting algorithm to study because it is one of the few phylogenetic
reconstruction methods that directly returns a rooted tree on a collection
of species.  One motivation for studying the UPGMA algorithm was the
work of Aldous \cite{Aldous}, where it was observed that rooted trees
that have been constructed from data do not typically have the same
underlying statistics as familiar speciation models such as the Yule 
process.  This raises the question of whether or not the Yule process 
is flawed, or the trees that have been constructed are biased because of
taxa selection, or inherent bias in the reconstruction methods. 
We believe that analyzing the partition of data space induced by a tree
reconstruction method can give some insight into the latter problem:
if regions corresponding to some tree shapes are inherently larger
than others, this indicates that the algorithm might favor those shapes
in the presence of noise or model misspecification of the equidistant 
assumption.

With these motivating problems in mind, we study the decomposition of
space induced by the UPGMA algorithm.
For a given
binary phylogenetic $X$-tree $T$ (that is, with leaf labels $X$ but without edge
lengths), the region of $\mathcal{P}(T) \subseteq \rr^{n(n-1)/2}_{\geq 0}$ of dissimilarity
maps for which the algorithm returns 
the phylogenetic $X$-tree $T$ is a union of finitely many polyhedral cones, 
one for each ranking
function of the interior nodes of $T$.  We give explicit polyhedral descriptions
of the cones including facet defining inequalities and extreme rays, for all $T$ and
all $n$.  In particular, each cone has $O(n^{3})$ facet defining inequalities
but exponentially many extreme rays.  
We compute the spherical volumes of the regions $\mathcal{P}(T)$ for $n \leq 7$. 
These volumes give a measure of the proportion of dissimilarity maps for which UPGMA returns
a given combinatorial type of tree.
In particular, our computations seem to 
indicate that highly unbalanced trees have small volume UPGMA cones compared to more
balanced trees.
Our computation of spherical volumes builds on the Monte Carlo strategy in \cite{NJSphere}.


\section{Ranked Phylogenetic Trees and the UPGMA Algorithm}

The UPGMA method is an agglomerative tree reconstruction method that takes as an input ${n \choose 2}$ pairwise distances between a set of taxa $X$ and returns a rooted equidistant tree metric  on $X$.  In this section, we review necessary background
on ranked phylogenetic trees and the lattice of set partitions as they
pertain to describing the UPGMA algorithm.  We refer the reader to
\cite{Felsenstein} and \cite{SS} for background on phylogenetics.

\begin{defn}
Let $X$ be a finite set.  A \emph{phylogenetic $X$-tree} is a tree $T$ with leaves 
bijectively labeled by the
set $X$.  A phylogenetic $X$-tree is rooted if it has a distinguished root node $\rho$.
It is \emph{binary} if every interior vertex that is not a leaf has degree $3$
except for the root $\rho$, which has degree $2$.    
\end{defn}

Throughout this paper, unless stated otherwise, we assume that a \textit{tree}
 $T$ \textit{on $n$ taxa} is a rooted binary phylogenetic $X$- tree where $X =  [n]$.  In a rooted binary phylogenetic $X$ tree, $\rho$ is not labeled by an element of $X$.
 
A vertex $v \in V(T)$ is a \textit{descendant} of $u \in V(T)$ if the path 
from $\rho$ to $v$ includes $u$.   This relation induces a partial order on the vertices of $T$ and we can write $u \leq_{T} v$.  Let $V^{\circ}$ denote the set of
interior (i.e. nonleaf) vertices of $T$.  
A \textit{rank function} on $T$ is a bijection 
$r : V^{\circ} \to \{ 1,2, ... , |V^{\circ}| \}$ satisfying 
$u \leq_{T} v \rightarrow r(u) \leq r(v)$. 
The number of rank functions on $T$ is :
$
{|V^{\circ}| ! }/{\prod_{v \in V^{\circ}} |\mathrm{de}(v)}|
$
where $\mathrm{de}(v)$ denotes the set of descendants of $v$ in the set $V^{\circ}$ \cite{StanleyVolumeI}.
Note that $v \leq_{T} v$, so that the number of descendants of $v$ will include $v$ itself. 
A tree with a rank function is called a \textit{ranked phylogenetic tree}.

The lattice of set partitions provides a useful alternate description of ranked phylogenetic trees.  See \cite{StanleyVolumeI} for
background and terminology for the theory of partially ordered sets.  
Let $\Pi_{n}$ consist of all partitions of a set with $n$ elements.
For simplicity, we identify this underlying set as $[n] = \{1,2, \ldots, n\}$.
Partitions are unordered, and consist of unordered elements.
The shorthand $A_{1} | \ldots | A_{k}$ denotes a partition
with $k$ parts.  For example $12|345$ is shorthand for the partition
$\{\{1,2\}, \{3,4,5\} \}$.

Partitions in $\Pi_{n}$ are ordered by refinement, so 
$A_{1}| \ldots |A_{k} \leq B_{1} | \ldots | B_{\ell}$  if and only
if for each $i \in [k]$ there exists a $j \in [\ell]$ satisfying $A_{i} \subseteq B_{j}$.
Every maximal chain in the lattice
of set partitions corresponds to a ranked phylogenetic tree.
Indeed, consider a maximal chain 
\[
C = 1|2| \cdots | n = \pi_{n} \lessdot \pi_{n-1} \lessdot \cdots \lessdot \pi_{2} \lessdot  \pi_{1} = 12 \cdots n 
\]
in $\Pi_{n}$. We use $\lessdot$ to denote a covering relation in the partial
order $\Pi_{n}$, and we use the convention that $\pi_{i}$ is always a partition with 
$i$ parts.

 Given $\pi_{i} \in C$, we write $\pi_{i} = \lambda^{i}_{{1}}| \lambda^{i}_{{2}} | \cdots  | \lambda^{i}_{{i}}$.  When $\pi_{i} \lessdot \pi_{i-1 }$, there are exactly two blocks $\lambda^{i}_{{j}}, \lambda^{i}_{{k}}$ that are joined in $\pi_{i-1 }$ but distinct in $\pi_{i}$.  If $v \in V^{\circ}$ where $r(v) = n - i$, then $\pi_{i - 1}$ joins the two blocks in $\pi_{i}$ that correspond to the subtrees of $T$ induced by the child nodes of $v$.  

The UPGMA algorithm constructs a rooted ranked phylogenetic $X$ tree from
a dissimilarity map $d$, as well as an equidistant tree metric $\delta$ 
which approximates $d$.  The algorithm works as follows:

\begin{alg}[UPGMA Algorithm]

\begin{itemize}

\medskip

\item Input: a dissimilarity map $d \in \rr^{ {n(n-1)/2}}_{\geq 0}$ on $X$.
\item Output: a maximal chain $C$ in the partition lattice $\Pi_{n}$ and an 
equidistant tree metric $\delta$.
\item Initialize $\pi_{n} = 1|2| \cdots | n$, and set $d^{n} = d$.
\item For $i = n-1, \ldots, 1$ do
\begin{itemize}
\item  From partition $\pi_{i+1} = \lambda^{i+1}_{1} | \cdots | \lambda^{i+1}_{i+1}$
and distance vector $d^{i+1} \in \rr^{(i+1)i/2}_{\geq 0}$
choose $j, k$ be so that $d^{i+1}(\lambda^{i+1}_{j}, \lambda^{i+1}_{k})$ is minimized.
\item  Set $\pi_{i}$ to be the partition obtained from $\pi_{i+1}$ by
 merging $\lambda^{i+1}_{j}$ and $ \lambda^{i+1}_{k}$ and leaving all other
 parts the same. Let $\lambda^{i}_{i} = \lambda^{i+1}_{j} \cup \lambda^{i+1}_{k}$.
\item  Create new distance $d^{i} \in \rr^{i(i-1)/2}_{\geq 0}$ by
$d^{i}(\lambda, \lambda') = d^{i+1}(\lambda, \lambda')$ if $\lambda, \lambda'$ are
both parts of $\pi_{i+1}$ and
$$
d^{i}(\lambda, \lambda^{i}_{i})
=  \frac{  |\lambda^{i+1}_{j}|}{ |\lambda^{i}_{i}|} 
d^{i+1}( \lambda, \lambda^{i+1}_{j} )  + 
\frac{|\lambda^{i+1}_{k}|}{ |\lambda^{i}_{i}|} 
d^{i+1}( \lambda, \lambda^{i+1}_{k} )
$$
\item  For each $x \in \lambda^{i+1}_{j}$ and $y \in\lambda^{i+1}_{k}$, 
set $\delta(x,y) = d^{i+1}(\lambda^{i+1}_{j}, \lambda^{i+1}_{k})$
\end{itemize}
\item Return:  Chain $C = \pi_{n} \lessdot \cdots \lessdot \pi_{1}$ and
equidistant metric $\delta$.
\end{itemize}
\end{alg}

Note that that step which recalculates distances, the weighted average 
$$
d^{i}(\lambda, \lambda^{i}_{i})
=  \frac{  |\lambda^{i+1}_{j}|}{ |\lambda^{i}_{i}|} 
d^{i+1}( \lambda, \lambda^{i+1}_{j} )  + 
\frac{|\lambda^{i+1}_{k}|}{ |\lambda^{i}_{i}|} 
d^{i+1}( \lambda, \lambda^{i+1}_{k} )
$$
is used to determine the new distance.  This is simply a computationally efficient
strategy to compute the average of distances
\begin{equation} \label{eq:averaged}
d^{i}(\lambda, \lambda')  =  \frac{1}{|\lambda| \cdot |\lambda'|} \sum_{x \in \lambda, y \in \lambda'} d(x,y)
\end{equation}
a formula we will make use of later.

\begin{ex}

Let $\mathbf{d} = (1,   2,    1.8,   1.7,  2,  2.6,  3.1,   2.4,  2.6,  1.2) \in \rr^{5(5-1)/2}_{\geq 0}$, be a dissimilarity map on $5$ taxa.
  
The UPGMA algorithm performs the following steps, where an underline is used to denote the 
smallest value in the present metric.

\[
\begin{matrix} 
12 &   13 & 14 & 15 & 23 & 24 & 25 & 34 & 35 & 45 \\  
(\underline{1}, &  2,   & 1.8, &  1.7, & 2, & 2.6, & 3.1, &  2.4, & 2.6, & 1.2) \end{matrix}
\]
\[
\begin{matrix} 
12, 3  & 12, 4  & 12, 5   & 34  & 35 & 45 \\ 
(2, & 2.2, & 2.4, & 2.4, & 2.6, & \underline{1.2}) 
\end{matrix}
\]
\[
\begin{matrix}
 12, 3 & 12, 45 & 3, 45 \\ 
 \underline{2} & 2.3 & 2.5 \end{matrix}
\]
\[
\begin{matrix} 123, 45 \\ \underline{2.367} \end{matrix}
\]
where 
\[
2.367  =  \left(\frac{ |12|} { |12| + |3|}\right) (2.3) + \left( \frac{ |3|} { |12| + |3|} \right) (2.5)
\]

The resulting rooted metric tree produced by the UPGMA algorithm is displayed in
Figure \ref{fig:treemetric}.

\begin{figure}[h]  \label{fig:treemetric}
\includegraphics[width=6cm]{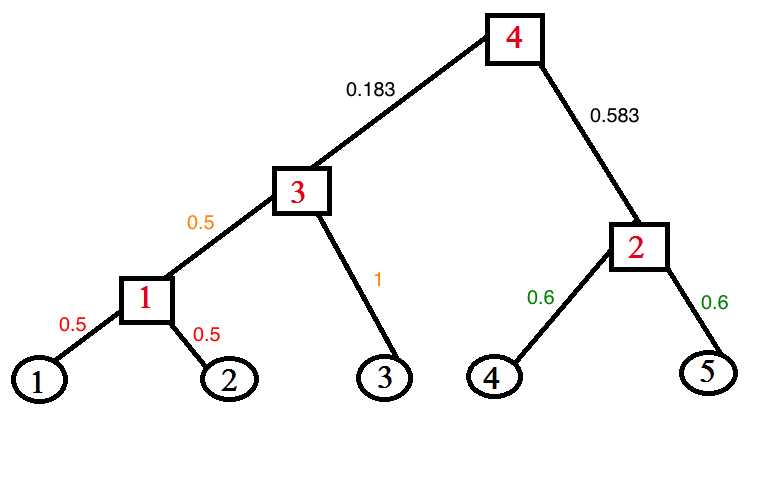} 
 \begin{caption}{The tree metric $\delta$}\end{caption}
  \end{figure}
The corresponding chain in the lattice of partitions $\Pi_{5}$ is
   \[
  C =    1|2|3|4|5 \lessdot 3|4|5 |12 \lessdot 3| 12 |45 \lessdot 45 |123  \lessdot 12345.
      \]
\end{ex}
 
\medskip



\section{UPGMA regions and UPGMA cones}

The UPGMA algorithm takes as input a dissimilarity map 
$d \in \mathbb{R}^{n(n-1)/2}_{\geq 0}$ and returns a rooted equidistant
tree metric.
If we ignore the
resulting metric tree that is output, and only record the rooted tree
computed at each step of the algorithm, the UPGMA algorithm
produces a rooted tree and a ranking function of the internal nodes
corresponding to precisely one maximal chain in the partition lattice.  Our goal is 
to understand the set of dissimilarity maps $d$, such that the UPGMA
returns a rooted tree $T$, or equivalently a given chain $C$ in the partition
lattice $\Pi_{n}$.  For a given leaf-labeled rooted tree $T$ let 
$\mathcal{P}(T) \subseteq \rr^{n(n-1)/2}_{\geq 0}$
denote the closure of the set of dissimilarity maps such that the UPGMA
algorithm returns $T$.  
The set $\mathcal{P}(T)$ is called the \textit{UPGMA region} associated to the tree $T$.
Similarly, for a maximal chain $C$ in $\Pi_{n}$, let $\mathcal{P}(C)\subseteq \rr^{n(n-1)/2}_{\geq 0}$ denote the 
closure of the
set of dissimilarity maps such that the UPGMA algorithm returns the chain
$C$.  


Our goal in this Section is to describe the sets $\mathcal{P}(T)$ and $\mathcal{P}(C)$.
Clearly
$\mathcal{P}(T)  =  \cup \mathcal{P}(C)$   
where the union is over all maximal chains in $\Pi_{n}$ whose associated
tree is $T$.

\begin{thm}\label{thm:main}
For each chain $C \in \Pi_{n}$ the 
set $ \mathcal{P}(C)$ is a pointed polyhedral cone.  The cone
has $O(n^{3})$ facet defining inequalities, and exponentially many
extreme rays.  Each covering relation in the chain $C$ determines a
collection of facet defining inequalities for $\mathcal{P}(C)$.  
Each element of the chain $C$ determines a collection of extreme
rays of $\mathcal{P}(C)$.
\end{thm}

We refer the reader to \cite{Ziegler} for background material on
polyhedral geometry.
To prove Theorem \ref{thm:main}, we will provide a more general
result for the description of cones associated to partial chains.
A partial chain $C$ is a sequence 
$$ \pi_{s} \lessdot \pi_{s-1} \lessdot \cdots \lessdot \pi_{t}$$
for some $n \geq  s \geq t \geq 1$.  
The fact that these are covering relations guarantees
that at each step, $\pi_{i+1} \lessdot \pi_{i}$ we are simply
joining a pair of parts together.  This means that any partial
chain $C$ can be intermediate information that is calculated 
between steps $s$ and $t$ of the UPGMA algorithm.  

For a partial chain
$C$, let $\mathcal{P}(C)$ denote the set of all dissimilarity maps
$d \in \rr^{s(s-1)/2}_{\geq 0}$ which the UPGMA algorithm could
produce on steps $s$ through $t$ of the algorithm.
The coordinates in the space $\rr^{s(s-1)/2}$ are the
$s(s-1)/2$ distances $d(\lambda^{s}_{j}, \lambda^{s}_{k})$.

\begin{prop}\label{prop:ineqs}
Let $C$ be a partial chain in $\Pi_{n}$.
Let $\mathcal{P}(C) \subseteq \rr^{s(s-1)/2}$ be the set
of dissimilarity maps for which steps $s$ through $t$  of the 
UPGMA algorithm return the partial chain $C$.  
For each covering relation $\pi_{i} \lessdot \pi_{i-1}$
let $\lambda^{i}_{j(i)}$ and $\lambda^{i}_{k(i)}$ be the pair of parts of $\pi_{i}$
that are joined in $\pi_{i-1}$.  Then $\mathcal{P}(C)$ is the solution
to the following system of linear inequalities:

$$d(\lambda^{s}_{j}, \lambda^{s}_{k}) \geq 0 \mbox{ for all } j,k$$
$$\mbox{ for } i = s, \ldots, t-1,  \mbox{  and for all pairs } j,k \neq j(i),k(i) $$
$$
\frac{1}{|\lambda^{i}_{j(i)}|  |\lambda^{i}_{k(i)}|}
\sum_{ \lambda^{s}_{j} \subseteq \lambda^{i}_{j(i)}, 
\lambda^{s}_{k} \subseteq \lambda^{i}_{k(i)}} |\lambda^{s}_{j}| |\lambda^{s}_{k}| d(\lambda^{s}_{j},\lambda^{s}_{k})  \leq
\frac{1}{|\lambda^{i}_{j}|  |\lambda^{i}_{k}|}
\sum_{ \lambda^{s}_{j} \subseteq \lambda^{i}_{j}, 
\lambda^{s}_{k} \subseteq \lambda^{i}_{k}} |\lambda^{s}_{j}| |\lambda^{s}_{k}| d(\lambda^{s}_{j},\lambda^{s}_{k})  
$$
\end{prop}

Note that if $s > t $  we only need the nonnegativity constraint $d(\lambda^{s}_{j(s)}, \lambda^{s}_{k(s)}) \geq 0$, as 
the other inequalities  $d(\lambda^{s}_{j}, \lambda^{s}_{k}) \geq 0$ follow from $d(\lambda^{s}_{j(s)}, \lambda^{s}_{k(s)}) \leq d(\lambda^{s}_{j}, \lambda^{s}_{k})$.

\begin{proof}
At step $i$ of the UPGMA algorithm, we choose the pair of $\lambda^{i}_{j(i)}$
and $\lambda^{i}_{k(i)}$ to merge such that $d^{i}(\lambda^{i}_{j(i)},\lambda^{i}_{k(i)})$
is minimized.  Using the formula 
$$
d^{i}(\lambda^{i}_{j},\lambda^{i}_{k}) =
\frac{1}{|\lambda^{i}_{j}|  |\lambda^{i}_{k}|}
\sum_{x \in \lambda^{i}_{j}, 
y \in \lambda^{i}_{k}}  d(x,y)
$$
twice shows that
$$
d^{i}(\lambda^{i}_{j},\lambda^{i}_{k}) =
\frac{1}{|\lambda^{i}_{j}|  |\lambda^{i}_{k}|}
\sum_{ \lambda^{s}_{j} \subseteq \lambda^{i}_{j}, 
\lambda^{s}_{k} \subseteq \lambda^{i}_{k}} |\lambda^{s}_{j}| |\lambda^{s}_{k}| d(\lambda^{s}_{j},\lambda^{s}_{k}).
$$
This yields precisely the inequalities in the statement of the proposition
at step $i$.
\end{proof}

\begin{prop}\label{prop:O(n^3)}
Given a maximal chain $C \in \Pi_{n}$, there are $O(n^{3})$ facet defining inequalities for $\mathcal{P}(C)$. 
\end{prop}

\begin{proof}
At step $t$, there are ${t \choose 2}$ ways to merge two blocks of $\pi_{t}$, and  the pair of parts $d(\lambda^{t}_{j(t)}, \lambda^{t}_{k(t)})$ merged at step $t$ can be paired with ${t \choose 2} - 1$ other pairs of parts.  So  ${t \choose 2} - 1$ new inequalities are introduced at step $t$.  An elementary identity for binomial coefficients tells us that for $a, b \geq 0$,  $\sum_{r = b}^{a} { r \choose b} = { a + 1 \choose b + 1}$. Thus there are
\[
\sum_{t = 2}^{n}  ({ t \choose 2} - 1) =  {n + 1 \choose 3} - n +1
\]
facet defining inequalities.
\end{proof}

Now we provide a description of the extremal rays of the cones of 
partial chains $\mathcal{P}(C)$, for partial chains starting with the bottom 
element $\pi_{n} = 1|2|\cdots | n$.  The polyhedral description 
of the cones $\mathcal{P}(C)$ for more general 
partial chains is used in the proof of the main cases of interest.

\begin{defn}
Given a partition $\pi_{k} =  \lambda_{1} | \lambda_{2} | \cdots | \lambda_{k} \in \Pi_{n}$ 
a \emph{traversal} of $\pi_{k}$ is a subset $F  \subset {[n] \choose 2}$ of size
${k \choose 2}$, where each element of $F$ is a pair $\{p,p'\} \in \pi$ satisfying $p \in \lambda, p' \in \lambda'$.  There is precisely one such pair $p, p'$ for every pair of parts $\lambda, \lambda'$ of $\pi_{k}$.
\end{defn}

For example, the partition $12|3|45$ has $2^{2}\cdot (2 \cdot 1) \cdot (2 \cdot 1) = 16$ traversals. 

\begin{defn}
Let $\pi_{k}  =  \lambda_{1} | \lambda_{2} | \cdots | \lambda_{k} \in \Pi_{n}$.  Let $F$ be a traversal of $\pi_{k}$.  The \textit{induced vector} of $F$, denoted $v(F)$, is the vector in $\rr^{n \choose 2}$ such that 

\begin{enumerate}

\item $v(F)_{ij } = 0$ if the pair $i,j$ is not in the traversal $F$.
\item $i,j \in F$, $v(F)_{ij} =|\lambda_{k(i)}||\lambda_{k(j)}|$ where $i \in \lambda_{k(i)}$ and $j \in \lambda_{k(j)}$.

\end{enumerate}

\end{defn}

Consider the traversal $\{ \{1,3\}, \{1,4\}, \{3,5\} \}$ of the partition $12|3|45$.  This traversal induces the vector $(0,2,4,0,0,0,0,0,2,0)$.

\begin{thm}\label{thm:extremerays}
Let $C = \pi_{n} \lessdot \pi_{n-1} \lessdot \cdots \lessdot \pi_{t}$ be a 
grounded partial chain in $\Pi_{n}$.  
Then $\mathcal{P}(C)$ is a cone with extreme rays given by the set of vectors
\begin{eqnarray*}
&  &  
\{ e(k,l) :  k, l \mbox{ are not in the same part of the partition } \pi_{t} \} \\
& \bigcup & 
\bigcup_{i = t-1}^{n}  \{ v(F) \ : \ F \ \text{is a traversal of $\pi_{i}$} \ \}
\end{eqnarray*}
\end{thm}

Note that $e(k,l)$ denotes the standard unit vector in $\rr^{n(n-1)/2}$
with a $1$ in the $k,l$ position and a $0$ elsewhere.

The remainder of this section consists of the proof of Theorem \ref{thm:extremerays}
and completes our description of the cones $\mathcal{P}(C)$.  The proof will
be broken into a number of pieces, and will work by induction on both $t$ and $n$.

Let ${\bf 1}^{t}$ denote the vector in $\rr^{t(t-1)/2}$ all of whose coordinates
are equal to one.  Note that ${\bf 1}^{n}$ is the induced vector of the single
traversal associated to the partition $1|2|\cdots | n$, which appears in 
every partial chain.

\begin{lemma}\label{lem:allones}
Let $C = \pi_{s} \lessdot \cdots \lessdot \pi_{t}$ be a partial chain in $\Pi_{n}$ 
with $s > t$.  Then 
\begin{enumerate}
\item 
${\bf 1}^{s}$ is an 
extreme ray of $\mathcal{P}(C)$ and
\item ${\bf 1}^{s}$  is the only 
extreme ray of $\mathcal{P}(C)$ that has a nonzero $(\lambda_{j(s)}^{s},\lambda_{k(s)}^{s})$ coordinate
where $(\lambda_{j(s)}^{s},\lambda_{k(s)}^{s})$ is the pair of parts  joined together
in the partition $\pi_{s-1}$. 
\end{enumerate}
\end{lemma}

\begin{proof}
First of all, all the inequalities of Proposition \ref{prop:ineqs} are satisfied
with equality by ${\bf 1}^{s}$ so that ${\bf 1}^{s} \in \mathcal{P}(C)$,
except for the single inequality $d(\lambda_{j(s)},\lambda_{k(s)}) \geq 0$, which is satisfied
strictly.  Hence the extreme ray ${\bf 1}^{s}$ is in the intersection of 
all the facet defining inequalities except for one.  Since $\mathcal{P}(C)$
is a pointed cone because it is contained in the positive orthant,
this implies that ${\bf 1}^{s}$ is an extreme ray.  This proves part (1).
Furthermore, since every extreme ray of a cone is the intersection
of some of its facet defining inequalities, every other extreme ray
must have the inequality $d(\lambda_{j(s)},\lambda_{k(s)}) \geq 0$ as an active inequality.
This proves part (2).
\end{proof}

Note that Lemma \ref{lem:allones} implies that if $s > t$, the 
vertex figure of $\mathcal{P}(C)$ is a pyramid with apex ${\bf 1}^{s}$.

Let $C = \pi_{s} \lessdot \cdots \lessdot \pi_{t}$ be a partial chain,
and $C'$ a partial chain obtained as a final segment of $C$, that
is, there is a $s < u \leq t$, such that
$C' = \pi_{u}\lessdot \cdots \lessdot \pi_{t}$.  The UPGMA
algorithm induces a natural linear map $A(C,C'): \rr^{s(s-1)/2} \rightarrow
\rr^{u(u-1)/2}.$  In particular, it is defined by
$$
(A(C,C') d) ( \lambda, \lambda')  =  \frac{1}{|\lambda| |\lambda'|}
\sum_{\mu, \mu' \in \pi_{s} \atop
\mu \subseteq \lambda, \mu' \subseteq \lambda'
}
|\mu| |\mu'| d(\mu,\mu')
$$
where $\lambda, \lambda'$ are parts of $\pi_{u}$.
Note, in particular, the quantity $d(\mu,\mu')$ only appears in the formula
for $(A(C,C') d) ( \lambda, \lambda')$, so that $A(C,C')$ is a
coordinate substitution map (Definition \ref{defn:coordsub})
when restricted to the coordinates $d(\mu, \mu')$
where $\mu, \mu'$ are in different parts of $\pi_{s}$.

With the preceding paragraph in mind, we
let $\tilde{\mathcal{P}}(C)$ denote the intersection of 
$P(C)$ with the hyperplane $\{ d: d(\lambda_{j(s)}, \lambda_{k(s)}) = 0 \}$.

\begin{prop}\label{prop:project}
Let $C = \pi_{s} \lessdot \cdots \lessdot \pi_{t}$ be a partial chain
and with final segment $C ' = \pi_{s-1} \lessdot \cdots \lessdot \pi_{t}$.  Then $A(C,C') : \tilde{\mathcal{P}}(C) \rightarrow \mathcal{P}(C')$
is surjective, and $\tilde{\mathcal{P}}(C) = A(C,C')^{-1}(\mathcal{P}(C')) \cap
\rr^{s(s-1)/2 -1}_{\geq 0}$.
\end{prop}

\begin{proof}
Note that by definition of the UPGMA algorithm, the map 
$A(C,C') : \mathcal{P}(C) \rightarrow \mathcal{P}(C')$ is 
surjective.  If a vector $d^{s} \in \mathcal{P}(C)$,
then so is the vector 
$$d' = d^{s} -  d^{s}(\lambda^{s}_{j(s)}, \lambda^{s}_{k(s)})
e( \lambda^{s}_{j(s)}, \lambda^{s}_{k(s)}),$$ 
obtained by zeroing
out the $( \lambda^{s}_{j(s)}, \lambda^{s}_{k(s)})$ coordinate.
However, $A(C,C')d^{s} = A(C,C')d'$, which implies that 
$A(C,C') : \tilde{\mathcal{P}}(C) \rightarrow \mathcal{P}(C')$
is surjective.

To see that $\tilde{\mathcal{P}}(C) = A(C,C')^{-1}(\mathcal{P}(C')) \cap
\rr^{s(s-1)/2 -1}_{\geq 0}$,
note that the inequalities that describe 
$\tilde{\mathcal{P}}(C)$ are precisely the pullbacks of the
inequalities that describe $\mathcal{P}(C')$, plus nonnegativity
constraints, since none of the inequalities on
$\mathcal{P}(C)$ coming from the covering relation
$\pi_{s} \lessdot \pi_{s-1}$ are needed.
\end{proof}

\begin{defn}\label{defn:coordsub}
A linear transformation $\phi: \rr^n \rightarrow \rr^m$ is a \emph{coordinate
substitution} if for each of the coordinate vectors 
$e_i$, $\phi(e_i) = c_i e_{\alpha(i)}$ with $c_i > 0$,
where $\alpha: [n] \rightarrow [m]$.  That is, each coordinate maps
to a scaled version of another coordinate.
\end{defn} 

\begin{lemma}\label{lemma:coordsub}
Let $D \subseteq \rr^m$ be a polyhedral cone, $\phi: \rr^n \rightarrow
\rr^m$ be a coordinate substitution with associated map $\alpha$, and $C \subseteq \rr^n$ a polyhedral
cone such that $\phi(C) = D$.  Suppose that $C = \rr^n_{\geq 0} \cap
\phi^{-1}(D)$.
Let $V$ be the set of extreme
rays of $C$.  Then extreme rays of $D$ consist of all
vectors obtained by the following procedure:

For each extreme ray $\sum_j a_{j}e_j \in V$, consider 
all vectors of the form $\sum_j a_j/c_{\beta(j)}  e_{\beta(j)}$ ranging over
all functions $\beta :[m] \rightarrow [n]$ such that $\alpha(\beta(j)) = j$ for 
all $j$.
\end{lemma}

\begin{proof}
It suffices to  show that under the hypotheses of the Lemma, every extreme
ray of $C$ maps onto an extreme ray of $D$.  Indeed, if that is the case,
the extreme rays of $C$ are precisely the vertices of the polytopes
$\phi^{-1}(v) \cap \rr^n_{\geq 0}$ as $v$ ranges over the extreme rays of
$V$.  Note that since $\phi$ is a coordinate substitution
$\phi^{-1}(v)$ is isomorphic to a product of simplices, the simplices
being defined over coordinate subsets over the form $\alpha^{-1}(j)$.
The vertices of these products of simplices have the form of the 
statement of the Lemma.

Hence, it suffices to show the claim that every extreme ray of $C$
maps onto an extreme ray of $D$.  So suppose that
$v'$ is an extreme ray of $C$ such that $\phi(v') = v$ is 
not an extreme ray of $D$.  Then there exists $w, u \in D$, not
equal to $v$ such that $v = w + u$.  Using these vectors, 
we construct $w', u' \in C$ not equal to $v'$ such that $v' = w' +u'$.
For each $i$ such that $\alpha(i) = j$ define
$$w'_i  =  \frac{w_j}{v_j} v'_i \quad \quad \mbox{and} \quad \quad
u'_i  =  \frac{u_j}{v_j} v'_i.$$
Clearly with this choice, we have $ v' = w' + u'$ since $v_j = w_j + u_j$, and both $w'$ and 
$u'$ consist of nonnegative vectors.  Also, since $w, u$ not equal $v$, neither are
$w', u'$ equal to $v'$.  So we must show
that $\phi(w') = w$ and $\phi(u') = u$.
But 
$$\phi(w')_j =  \sum_{i: \alpha(i)=j}  \frac{w_j}{v_j} c_i =  
\frac{w_j}{v_j} \sum_{i: \alpha(i)=j} c_i
 = \frac{w_j}{v_j}  v_j = w_j.$$
Similarly for $u'$, which completes the proof.
\end{proof}

We now have all the ingredients to prove Theorem \ref{thm:extremerays}.

\begin{proof}[Proof of Theorem \ref{thm:extremerays}]
Let $C = \pi_{s} \lessdot \cdots \lessdot \pi_{t}$.
First of all, note that if $s = t$, then $\mathcal{P}(C)$ is the 
positive orthant in $\rr^{s(s-1)/2}$, whose extreme rays are the standard
unit vectors.

Now assume that $s > t$.  According to Lemma \ref{lem:allones}, the vector 
${\bf 1}^{s}$ is an extreme ray of $\mathcal{P}(C)$.
Letting $C' = \pi_{s-1} \lessdot \cdots \lessdot \pi_{t}$,
Proposition \ref{prop:project}
we see that all other extreme rays of $\mathcal{P}(C)$ can
be obtained by applying Lemma \ref{lemma:coordsub} to the
extreme rays of $\mathcal{P}(C')$.  Repeating this procedure
for the extreme rays of $\mathcal{P}(C)$ that do not 
map to ${\bf 1}^{s-1} \in \mathcal{P}(C')$,
we see that every extreme ray of 
$\mathcal{P}(C)$ besides ${\bf 1}^{s}$ can be obtained
as a vertex of $A(C, C_{u})^{-1}({\bf 1}^{u})$ where
$C_{u} = \pi_{u} \lessdot \cdots \lessdot \pi_{t}$,
plus the vertices of
$A(C, C_{t})^{-1}(e(\lambda_{k},\lambda_{l}))$.

To complete the proof of the theorem
we must analyze the vertices of $A(C, C_{t})^{-1}(e(\lambda_{k},\lambda_{l}))$
and show that the vertices
of
$A(C, C_{u})^{-1}({\bf 1}^{u})$
are precisely the induced vectors  from the traversals of
$\pi_{u}$. For both of these statements, we can use
Lemma \ref{lemma:coordsub}.

Indeed, $A(C, C_{u})$ is the map such that
$$(A(C,C_{u})d)( \lambda, \lambda')
=  \frac{1}{|\lambda| \cdot |\lambda'|}\sum_{x \in\lambda \atop y \in \lambda'} d(x,y).
$$
This implies, by Lemma \ref{lemma:coordsub} that the vertices of 
$$A(C, C_{t})^{-1}(e(\lambda,\lambda'))$$
are $|\lambda| \cdot |\lambda'| e(k,l)$ such that
$k \in \lambda$ and $l \in \lambda'$.  Since
we can ignore the scaling factor $|\lambda| \cdot |\lambda'|$
when describing extreme rays, 
taking the union over 
all pairs $\lambda, \lambda' \in \pi_{t}$, yields the set
of rays $\{ e(k,l) :  k, l \mbox{ are not in the same part of the partition } \pi_{t} \}$ from Theorem \ref{thm:extremerays}.

Similarly, applying Lemma \ref{lemma:coordsub} to 
the map $A(C,C_{u})$ and the vector ${\bf 1}^{u}$
yields the set of induced vectors $v(F)$ associated
to the partition $\pi_{u}$.
Indeed, the coordinate $1$ in the $(\lambda, \lambda')$
position of ${\bf 1}^{u}$ produces
an entry of $|\lambda| \cdot |\lambda'|$
in exactly  on of the positions $d(x,y)$ such that $x \in \lambda,
y \in \lambda'$.  This completes the proof of Theorem \ref{thm:extremerays}.  
\end{proof}

We now show that Theorem \ref{thm:extremerays} implies that the UPGMA cones have exponentially many extreme rays.

\begin{prop}\label{prop:exponentiallymany} The cones $\mathcal{P}(C)$ have exponentially many extreme rays.
\end{prop}

\begin{proof}
Given $\pi_{s} = \lambda_{1}^{s} | \cdots | \lambda_{s}^{s}$, the number of traversals is the product of the pairwise products of the cardinalities of the blocks of $\pi_{s}$.  So the number of extreme rays induced by $\pi_{s}$ is 
\[
\prod_{\{i,j\} \subset { [s] \choose 2} }|\lambda_{i}^{s}| | \lambda_{j}^{s}| = \prod_{i = 1}^{s} |\lambda_{i}^{s}|^{s-1}.
\]
Given a maximal chain $C \in \Pi_{n}$, the total number of extreme rays will be 
\[
\sum_{s = 2}^{n}  \prod_{i = 1}^{s} |\lambda_{i}^{s}|^{s-1}
\]
which is exponential.

\end{proof}

Note that Propositions \ref{prop:ineqs}, \ref {prop:O(n^3)}, \ref{prop:exponentiallymany}  and Theorem \ref{thm:extremerays} yield Theorem \ref{thm:main}.


\section{Applications of Theorem \ref{thm:extremerays}}

We use the characterization of the extreme rays
of the cones $\mathcal{P}(C)$ to provide easy geometric applications.  First of
all, in general, the set $\mathcal{P}(T)$ of
all dissimilarity maps for which UPGMA returns a given tree,
is not a convex set in general.  Second,
the partition of the positive orthant into the cones
$\mathcal{P}(C)$ does not have the structure of a 
polyhedral fan, which means cones do not intersect in 
their boundary in an especially nice way.  Thirdly, we show the comb tree topology minimizes the number of rays in a UPGMA cone.

\begin{cor}\label{cor:notconvex}
The UPGMA regions $\mathcal{P}(C)$ are not convex in general.
\end{cor}

\begin{proof}
We give an example for $n = 4$.  Let  $T = ((12)(34))$.  Then $\mathcal{P}(T) = \mathcal{P}(C_{1}) \cup \mathcal{P}(C_{2})$ where
\[
C_{1} = 1|2|3|4 \lessdot 3|4|12 \lessdot 12|34 \lessdot 1234
\]
\[
C_{2} =  1|2|3|4 \lessdot 1|2|34 \lessdot 34|12 \lessdot 1234
\]
Now $v_{1} = (0,0,2,2,0,1)$ is an extreme ray of $P(C_{1})$ induced by a traversal of $3|4|12$ and $v_{2} = (1,0,2,2,0,0)$ is an extreme ray of $\mathcal{P}(C_{2})$ induced by a traversal of $1|2|34$.  Let $d$ be the convex combination 
\[
d = \frac{1}{2}v_{1} + \frac{1}{2}v_{2} = \left(\frac{1}{2}, 0, 2, 2, 0, \frac{1}{2} \right)
\]
If $d$ is input into UPGMA, the algorithm will return a tree with either $(1,3)$ or $(2,4)$ as a cherry, so $d$ is not in $\mathcal{P}(T)$.  So, in general, UPGMA regions are not convex unless $\mathcal{P}(T) = \mathcal{P}(C)$ for a single chain $C$ in $\Pi_{n}$.   
\end{proof}

  A \textbf{fan} is a family $\mathcal{F}$  of cones in $\rr^{n}$ such that 
        \begin{enumerate}
        \item if $P \in \mathcal{F}$ then every nonempty face of $P$ is in $\mathcal{F}$
        \item if $P_{1}, P_{2} \in \mathcal{F}$ then $P_{1} \cap P_{2} \in \mathcal{F}$.
        \end{enumerate}
        
        \addvspace{2pc}
  \begin{cor}\label{cor:notfan} The UPGMA cones do not partition $\rr^{n \choose 2}$ into a fan.
\end{cor}

\begin{proof}      
        Consider the two chains in $\Pi_{4}$
       \[
       C_{1} = \quad 1|2|3|4 \lessdot 3|4 |12 \lessdot 4| 123 \lessdot 1234
       \] 
          \[
       C_{2} = \quad 1|2|3|4 \lessdot 2|4 |13 \lessdot 4| 123 \lessdot 1234
       \] 
       
       \addvspace{3pc}
 
 \par
The vector $(0,0,0,1,1,1)$ generates an extreme ray of $P(C_{1}) \cap P(C_{2})$.   If $P(C_{1}) \cap P(C_{2})$ was a face of $P(C_{1})$ and $P(C_{2})$, then $(0,0,0,1,1,1)$ would generate a ray of $P(C_{1})$ and $P(C_{2})$.  However by Theorem \ref{thm:extremerays},  extreme rays of $P(C_{1})$ and $P(C_{2})$ must correspond to partitions in $\Pi_{4}$.  Only partitions with 3 blocks induce vectors with 3 nonzero coordinates, and no partition of the set $[4]$ has 3 blocks of equal cardinality. So, no traversal of a partition in $\Pi_{4}$ induces a multiple of $(0,0,0,1,1,1)$. Therefore the UPGMA cones are not a fan.  
\end{proof}

\begin{cor}\label{cor:comb}
For each $n$, the comb tree topology minimizes the number of extreme rays over all UPGMA cones in $\rr^{n \choose 2}$
\end{cor}

\begin{proof}
Fix $n$.  We will show that for each $1 \leq s \leq n$, the partitions whose parts have cardinalities $1,1, ..., 1,n - s+ 1$ minimize the number of traversals for all partitions with $s$ parts.  For all integers $x,y > 0$, we have $xy \geq (x + y - 1) (1)$.  
So for $\pi_{s} = \lambda_{1}^{s} | \cdots | \lambda_{s}^{s}$, the number of extreme rays induced by $\pi_{s}$ satisfies
\[
\prod_{\{i,j\} \subset { [s] \choose 2} }|\lambda_{i}^{s}| | \lambda_{j}^{s}| \geq \prod_{\{i,j\} \subset { [s] \choose 2} } (1) (|\lambda_{i}^{s}| + |\lambda_{j}^{s}| -1)
\]
The only type of partition in $\Pi_{n}$ with $s$ parts such that all pairs $\{i,j\} \subset { [s] \choose 2} $ satisfy either $|\lambda_{s}^{i}| = 1$ or $|\lambda_{s}^{j}| = 1$ is the type with $s-1$ singleton parts and one part of size $n - s + 1$.  Therefore partitions of this type minimize the number of associated induced vectors.  
\par
If $C$ is a maximal chain in $\Pi_{n}$ such that every $\pi_{s} $ in $C$ is of this type, then the tree returned by $d \in \mathcal{P}(C)$ has the comb tree topology.  Therefore this tree topology minimizes the number of extreme rays for  the cone $\mathcal{P}(C)$.  
\end{proof}

\section{Spherical Volumes of UPGMA Regions}
A natural way to measure the relative proportion of the region of dissimilarity maps $\mathcal{P}(T)$ returning the tree $T$ in the positive orthant returning a tree is to calculate the ${{n \choose 2} -1}$ dimensional measure of the surface arising as the intersection of the cones $\mathcal{P}(C) \subset \mathcal{P}(T)$ with the unit sphere $S$ in $\rr^{n \choose 2}$.  We refer to this measure as \textit{spherical volume}.

We estimated the spherical volume of UPGMA cones in two ways using Mathematica, polymake \cite{Polymake}, and the software \cite{Huggins}.  For the first method, we sampled points from the positive orthant using a spherical distribution and input the samples into UPGMA, recording which tree the algorithm returned on the input point. The volume of $\mathcal{P}(T)$ is then the fraction of the total sample points returning $T$.  We calculate volumes for $n = 4,5,6,7$ using this method.   

For the second method, we used a Monte Carlo strategy to estimate the surface area of the cones.  For $n = 4,5,6$, we used the software \cite{Huggins} for $n = 4,5,6$.  This software requires as input triangulations of point configurations that we computed using polymake \cite{Polymake}.   For $n = 7$, some triangulations for maximal chains in $\Pi_{7}$ were too large to compute and use.   We used Mathematica to implement a modification of the sampling strategy employed in \cite{NJSphere} along with the UPGMA algorithm.

The basic strategy using Monte Carlo integration to compute
spherical volumes can be described as follows. 
Given a simplicial cone ${\rm cone}(V)$ spanned by vectors $V = v_{1}, \ldots, v_{n}$,
it is easy to generate uniform samples from the
simplex ${\rm conv}(V)$.  The map that takes a point $x \in {\rm conv}(V)$
onto ${\rm cone}(V) \cap S$ is simply $x \rightarrow x/ \|x \|_{2}$.
The spherical volume is then the average value of the Jacobian
of this map.  To calculate the spherical volume of a
cone
$\mathcal{P}(C)$ of a full chain in situations
where we could only compute a triangulation of a cone
from a partial chain $\mathcal{P}(C')$, 
we generate random points from the partial cone $\mathcal{P}(C')$
and compute the average of the product of Jacobian and the indicator
function of lying in the cone $\mathcal{P}(C)$.

We summarize
the results here of those computations for $n = 4,5,6,7$ leaf trees,
only displaying results for the regions $\mathcal{P}(T)$.
In the tables below, we give estimates of the spherical volumes of the regions $\mathcal{P}(T)$.   The column Tree gives the tree in Newick format.  The column \#Chains refers to the number of cones producing the given tree.  The Volume column gives the total volume of all of the cones associated to the given tree, and the Fraction of Orthant column gives the portion of the positive orthant in $\rr^{n \choose 2}$ that returns the given tree \emph{topology} under UPGMA.  

Recall that $\mathcal{P}(T) = \cup \mathcal{P}(C)$ where $C$ ranges over the chains in $\Pi_{n}$ corresponding to $T$.  So, the number of cones associated to a tree $T$ depends on the number of rank functions that $T$ admits.  For example, in the table for $n = 5$, the tree  $T_{2} = (((12)3)(45))$ has 
$4!/(4 \cdot 2 \cdot 1 \cdot 1) = 3$ rank functions, and there are 3 cones in $\mathcal{P}(T_{2})$.  
\par
A more detailed explanation of the volume computations, as well as software and input files, is available at \cite{website}.

\small

 \begin{center}
\begin{tabular}{ | c | p{4cm} | c | c | p{2cm} |}
\hline
  & Tree & \# Chains & Volume &  Fraction of Orthant \\ \hline
 1 & $(((12)3)4)$   & 1 & 0.0238 & 0.5895 \\ \hline
 2 & $ ((12)(34)) $    & 2 & 0.0662 & 0.4099 \\ 
 \hline
\end{tabular} 
\medskip

\end{center}

\bigskip


\begin{center}
\begin{tabular}{ | c | p{4cm} |  c | c | p{2cm} |}
\hline
 & Tree  & \# Chains & Volume &  Fraction of Orthant \\ \hline

  1 & $ ((((12)3)4)5) $   & 1 & $ 8.57 \times 10^{-5}$ & 0.206 \\ 
 \hline
  2 & $(((12)3)(45))$    & 3 & $5.01 \times 10^{-4}$ & 0.604 \\ \hline
  3 & $(((12)(34))5)$    & 2 & $3.14 \times 10^{-4}$ &  0.189 \\ \hline
  \end{tabular}
\end{center}

\bigskip

\begin{center}
\begin{tabular}{ | c | p{4cm}   | c | c | p{2cm} |}
\hline
 & Tree &  \# Chains  & Volume & Fraction of Orthant \\ \hline
1 & (((((12)3)4)5)6)    & 1 & $ 2.05 \times 10^{-8}$ &  0.042 \\ \hline 
2 & ((((12)3)4)(56))    & 4  & $ 2.10 \times 10^{-7}$ & 0.216 \\  \hline  
 3 & ((((12)3)(45))6)    & 3  & $ 2.16*10^{-7}$ & 0.223 \\ \hline
  4 & (((12)3)((45)6))   & 6  & $ 4.5 \times 10^{-7}$ & 0.229 \\ \hline
5 & ((((12)(34))5)6)   & 2  & $1.05 \times 10^{-7}$ &  0.054 \\ \hline

6 & (((12)(34))(56))    & 8  & $ 9.06 \times 10^{-7}$ & 0.231  \\ \hline

\end{tabular}
\end{center}

\bigskip

\begin{center}
\begin{tabular}{ | c | p{4cm}   | c | c | p{2cm} |}
\hline
 & Tree &  \# Chains  & Volume & Fraction of Orthant \\ \hline
1 & ((((((12)3)4)5)6)7)    & 1 & $ 2.75 \times 10^{-13}$ &  0.0050 \\ \hline 
2 & (((((12)3)4)5)(67))  & 5  & $ 4.82 \times 10^{-12}$ & 0.0435 \\  \hline  
 3 & (((((12)3)4)(56))7)   & 4  & $ 6.32 \times10^{-12}$ & 0.0570 \\ \hline
  4 & ((((12)3)4)((56)7))   & 10  & $ 1.95 \times 10^{-11}$ & 0.1762 \\ \hline
5 &   (((((12)3)(45))6)7) & 3  & $4.45 \times 10^{-12}$ &  0.0402 \\ \hline
6 & ((((12)3)(45))(67))    & 15  & $ 5.72  \times 10^{-11}$ & 0.2581  \\ \hline
7 & ((((12)3)((45)6))7)    & 6  & $1.66  \times 10^{-11}$ & 0.0747  \\ \hline
8 & (((12)3)((45)(67)))    & 20 & $9.00  \times 10^{-11}$ & 0.2030  \\ \hline
9 & (((((12)(34))5)6)7)   & 2 & $1.73  \times 10^{-12}$ & 0.0078  \\ \hline
10 & ((((12)(34))5)(67))  & 10 & $2.63  \times 10^{-11}$ & 0.0593  \\ \hline
11 & ((((12)(34))(56))7)  & 8 & $3.33  \times 10^{-11}$ &0.0753  \\ \hline

\end{tabular}
\end{center}

\normalsize

The computations suggest some observations which might hold true
for large $n$.  As we have shown in Corollary \ref{cor:comb},
the cone associated to the single rank function on the comb tree
yields the cone $\mathcal{P}(C)$ with the fewest number of
extreme rays.  Our computations up to $n = 7$ suggests that this 
is also the cone with the smallest spherical volume.
See \cite{website} for those values.  The size of the
region $\mathcal{P}(T)$ appears be roughly proportional 
to the number of chains $C$ that yield the tree $T$
and appears to be smallest for the comb tree.
Furthermore, the relative proportion of the positive
orthant taken up by the comb tree topology appears to
be the smallest.  We predict that these patterns hold
for larger number of taxa as well.


\section*{Acknowledgments}

Ruth Davidson was partially supported by       the US National Science 
Foundation (DMS 0954865).
Seth Sullivant was partially supported by the David and Lucille Packard 
Foundation and the US National Science Foundation (DMS 0954865).

\end{document}